\documentclass{lmcs} 
\usepackage[utf8]{inputenc}
\pdfoutput=1

\usepackage{lastpage}
\lmcsdoi{19}{4}{10}
\lmcsheading{}{\pageref{LastPage}}{}{}%
{Jun.~25,~2021}{Nov.~09,~2023}{}


\usepackage{amsmath,amssymb,latexsym,stmaryrd}
\usepackage{datetime}
\usepackage{mathtools}
\usepackage{tikz-cd}
\usepackage{bbm}
\usepackage[parfill]{parskip}
\usepackage{hyperref}
\usepackage{graphicx}

\usetikzlibrary{cd}
\usetikzlibrary{decorations.pathmorphing}

\makeatletter
\newcommand{\singlespacing}{\let\CS=\@currsize\renewcommand{\baselinestretch}{1}\small\CS}
\newcommand{\doublespacing}{\let\CS=\@currsize\renewcommand{\baselinestretch}{1.75}\small\CS}
\newcommand{\normalspacing}{\let\CS=\@currsize\renewcommand{\baselinestretch}{\BLS}\small\CS}
\makeatother

\newcommand{\x}{\times}

\newcommand{\mes}{\mathbf{Mes}}

\newcommand{\dee}[1]{\mathrm{d}#1}

\newcommand{\average}{\mathbf{\mu}}
\newcommand{\monid}{\mathbf{\eta}}
\newcommand{\catname}[1]{{\mathbf{#1}}}

\newcommand{\Sbs}{\mbox{$\catname{SbStat}$}}
\newcommand{\Sb}{\mbox{$\catname{StBor}$}}
\newcommand{\Fins}{\mbox{$\catname{FinStat}$}}
\newcommand{\monoid}{$\mathtt{[0,\infty]}$}
\newcommand{\mon}{\; \tilde{\circ} \;}
\newcommand{\monf}{\; \tilde{\circ}_{fin} \;}
\newcommand{\relentf}{$RE_{fin}$}

\theoremstyle{plain}

\begin{document}

\title[A characterization of relative entropy on standard Borel spaces]{A categorical characterization of relative\texorpdfstring{\\}{ }entropy on standard Borel spaces}
\titlecomment{{\lsuper*}An extended abstract of this work appeared in the
  $33^{rd}$ Mathematical Foundations of Programming Semantics Conference.}

\author[N.~Gagn\'e]{Nicolas Gagn\'e\lmcsorcid{0009-0004-0207-0106}}[a]	
\author[P.~Panangaden]{Prakash Panangaden\lmcsorcid{0000-0001-8763-6172}}[b]

\address{D\'epartement d'informatique et de recherche op\'erationnelle \\ Universit\'e de Montr\'eal \\
    Montr\'eal, Qu\'ebec, Canada\vspace*{-\parskip}}		
\email{nicolas.gagne.5@umontreal.ca}  

\address{School of Computer Science \\ McGill University \\ Montr\'eal, Qu\'ebec, Canada\vspace*{-\parskip}}	
\email{prakash@cs.mcgill.ca}  

\thanks{This research has been supported by NSERC and by Google.}	

\begin{abstract}
\noindent We give a categorical treatment, in the spirit of Baez and Fritz, of
relative entropy for probability distributions defined on standard Borel spaces.
We define a category called \Sbs{} suitable for reasoning about
statistical inference on standard Borel spaces.  We define relative entropy as a
functor into Lawvere's category $[0,\infty]$ and we show convexity, lower
semicontinuity and uniqueness.
\end{abstract}

\maketitle

\section{Introduction}

The inspiration for the present work comes from two recent developments.
The first is the beginning of a categorical understanding of Bayesian
inversion and learning~\cite{Danos15,Dahlqvist16,Clerc17,dahlqvist2018borel}. The second is a
categorical reconstruction of relative
entropy~\cite{Baez11,Baez14,Leinster}.  The present paper provides a
categorical treatment of entropy in the spirit of Baez and Fritz in the
setting of standard Borel spaces, thus setting the stage to explore the role of
entropy in learning.

Recently there have been some exciting developments that bring some
categorical insights to probability theory and specifically to learning
theory.  These are reported in some recent papers by Clerc, Dahlqvist,
Danos and Garnier~\cite{Danos15,Dahlqvist16,Clerc17}.  The first of these
papers showed how to view the Dirichlet distribution as a natural
transformation thus opening the way to an understanding of higher-order
probabilities, while the second gave a powerful framework for constructing
several natural transformations.  In~\cite{Danos15} the hope was expressed
that one could use these ideas to understand Bayesian inversion, a core
concept in machine learning.  In ~\cite{Clerc17} this was realized in a
remarkably novel way.  These papers carry out their investigations in the
setting of standard Borel spaces and are based on the Giry
monad~\cite{Giry81,Lawvere64a}.

In ~\cite{Baez11,Baez14} a beautiful treatment of relative entropy is
given in categorical terms.  The basic idea is to understand entropy in
terms of the results of experiments and observations.  How much does one
learn about a probabilistic situation by doing experiments and observing
the results?  A category is set up where the morphisms capture the
interplay between the original space and the space of observations.  In
order to interpret the relative entropy as a functor they use Lawvere's
category which consists of a single object and a morphism for every
extended positive real number~\cite{Lawvere73}.  

Our contribution is to develop the theory of Baez et al.\@ in the setting
of standard Borel spaces; their
work is carried out with finite sets.  While the work of~\cite{Baez14}
gives a firm conceptual direction, it gives little guidance in the actual
development of the mathematical theory.  We had to redevelop the
mathematical framework and find the right analogues for the concepts
appropriate to the finite case.

\section{Background}
In this section we review some of the background.  We assume that the
reader is familiar with concepts from topology and measure theory as well
as basic category theory.  We have found books by Ash~\cite{Ash72},
Billingsley~\cite{Billingsley95} and Dudley~\cite{Dudley89} to be useful.

We will use letters like $X,Y,Z$ for measurable spaces and
capital Greek letters like $\Sigma,\Lambda,\Omega$ for $\sigma$-algebras.
We will use $p,q,\ldots$ for probability
measures.  Given $(X,\Sigma)$ and $(Y,\Lambda)$ and a measurable function
$f:X\to Y$ and a probability measure $p$ on $(X,\Sigma)$ we obtain a measure on
$(Y,\Lambda)$ by $p \circ f^{-1}$; this is called the \emph{pushforward}
measure or the \emph{image} measure.  

\subsection{The Giry monad}
We denote the category of measurable spaces and measurable functions by
$\mes{}$.  We recall the Giry~\cite{Giry81} functor
$\Gamma : \mes{} \rightarrow \mes{}$ which maps each measurable space $X$
to the space $\Gamma(X)$ of probability measures over $X$.  Let
$A\in\Sigma$, we define $\mathsf{ev}_A:\Gamma(X)\to[0,1]$ by
$\mathsf{ev}_A(p) = p(A)$.  We endow $\Gamma(X)$ with the smallest
$\sigma$-algebra making all the $\mathsf{ev}$'s measurable.  A morphism
$f:X\to Y$ in $\mes{}$ is mapped to $\Gamma(f): \Gamma (X)\to \Gamma(Y)$ by
$\Gamma(f)(p) = p \circ f^{-1}$.  With the following natural
transformations, this endofunctor is a monad: the Giry monad.  The natural
transformation $\monid:I\to \Gamma$ is given by $\monid_X(x)=\delta_x$, the
Dirac measure concentrated at $x$.  The monad multiplication
$\average : \Gamma^2 \rightarrow \Gamma$ is given by
 \[\forall A\in \mathcal{B}(X),\;\average_X(p)(A) := \int_{\Gamma(X)} \mathsf{ev}_A \; \dee{p}\]
 where $p$ is a probability measure in $\Gamma(\Gamma(X))$ and
 $\mathsf{ev}_A:\Gamma(X)\to [0,1]$ is the measurable function on $\Gamma(X)$
 defined by $\mathsf{ev}_A(p) = p(A)$.

 Even if $\mes{}$ is an interesting category in and of itself, the need for
 regular conditional probabilities forces us to restrict ourselves to a
 subcategory of standard Borel spaces.

\subsection{Standard Borel spaces and disintegration}

The Radon-Nikodym theorem is the main tool used to show the existence of
conditional probability distributions, also called Markov kernels, see the
discussion below.  It is a very general theorem, but it does not give as
strong regularity features as one might want.  A stronger theorem is
needed; this is the so-called \emph{disintegration theorem}.  It requires
stronger hypotheses on the space on which the kernels are being defined.  A
category of spaces that satisfy these stronger hypotheses is the category
of standard Borel spaces.  In order to define standard Borel spaces, we
must first define Polish spaces.

\begin{defi}
A \emph{Polish space} is a separable, completely metrizable topological
space.
\end{defi}

\begin{defi}
A \emph{standard Borel space} is a measurable space obtained by forgetting
the topology of a Polish space but retaining its Borel algebra.  The
category of standard Borel spaces has measurable functions as morphisms; we
denote it by $\Sb$.
\end{defi}

We can now state a version of the \emph{disintegration theorem}.  The
following is also known as \emph{Rohlin's disintegration
  theorem}.
\begin{thmC}[\cite{Rokhlin49}]\label{Dis}
  Let $(X,p)$ and $(Y,q)$ be two standard Borel spaces equipped with probability
  measures, where $q$ is the pushforward measure $q := p \circ f^{-1}$ for
  a Borel measurable function $f: X \rightarrow Y$.  Then, there exists a
  $q$-almost everywhere uniquely determined family of probability measures
  $\{p_y\}_{y \in Y}$ on $X$ such that
\begin{enumerate}
\item the function $y \mapsto p_y(A)$ is a Borel-measurable function for
  each Borel-measurable set $A \subset X$; 
\item $p_y$ is a probability measure on  $f^{-1}(y)$ for $q$-almost all $y \in Y$;
\item for every Borel-measurable function $h: X \rightarrow [0,\infty]$, \[
    \int_X h \; \dee p  = \int_Y \int_{f^{-1}(y)} h \; \dee p_y \dee q.\] 
\end{enumerate}
\end{thmC}
  The objects obtained are often called \emph{regular conditional
    probability distributions}.
One can find a crisp categorical formulation of disintegration in~\cite[Theorem 1]{Clerc17}.

\subsection{The Kleisli category of \texorpdfstring{\( \Gamma \)}{Gamma} on \texorpdfstring{\Sb{}}{StBor}}

It is well known that the Giry monad on $\mes{}$ restricted to \Sb{} admits the same monad structure.~\cite{Giry81}

The Kleisli category of $\Gamma$ has as objects standard Borel spaces and as
morphisms maps from $X$ to $\Gamma(Y)$: $h:X\to (\mathcal{B}_Y\to [0,1])$ which are
measurable.  Here $\mathcal{B}_Y$ stands for the Borel sets of $Y$ and $\Gamma(Y)$ has
the $\sigma$-algebra described above.  Now we can curry this to write it as
$h:X\x\mathcal{B}_Y\to[0,1]$ or $h(x,U)$ where $x$ is a point in $X$ and $U$ is a
Borel set in $Y$.  Written this way it is called a Markov kernel and
one can view it as a transition probability function or conditional
probability distribution given $x$.  Composition of morphisms $f:X\to Y$
and $g:Y\to Z$ in the Kleisli category is given by the formula
\[ (g\circ f)(x,V \in \mathcal{B}_{Z}) = \int_Y g(y,V) \; \dee{f(x,\cdot)}.     \]

For an arrow $s : Y \rightarrow \Gamma(X)$ in $\Sb$, we write $s_y$ for $s(y)$ or, in
kernel form $s(y,\cdot)$.  For arrows $t : Z \rightarrow \Gamma(Y)$ and
$s : Y \rightarrow \Gamma(X)$ in $\Sb$, we denote their Kleisli composition by
$s \mon t := \average_X \circ \Gamma(s) \circ t$.  For standard Borel spaces equipped with a
probability measure $p$, we sometimes omit the measure in the notation,
\emph{i.e.}\, we sometimes write $X$ instead of $(X,p)$.  We say a probability measure $p$ is
\emph{absolutely continuous} with respect to another measure $q$ on the
same measurable space $X$, denoted by $p \ll q$, if for all measurable
sets $B$, $q(B)=0$ implies that $p(B)=0$.  

We note that absolute continuity is preserved by Kleisli composition; the
proof is straightforward.
\begin{prop}\label{giry_abs}
Given a standard Borel space $Y$ with probability measures $q$ and $q'$ such that
$q \ll q'$.  Then, for arbitrary standard Borel space $X$ and morphism $s$ from $Y$
to $\Gamma(X)$, we have $s \mon q \ll s \mon q'$.  
\end{prop}

\section{The categorical setting}
In this section, following Baez and Fritz~\cite{Baez14} (see also
~\cite{Baez11}) we describe the category \Fins{} which
they use for their characterization of entropy on finite spaces.  We
then introduce the category $\Sbs{}$ which will be the arena for the
generalization to standard Borel spaces.

Before doing so, we define the notion of coherence which will play an
important role in what follows.

\begin{defi}
   Given standard Borel spaces $X$ and $Y$ with probability measure $p$ and $q$, respectively, a pair $(f,s)$, with $f : (X,p) \rightarrow (Y,q)$ and $s : Y \rightarrow \Gamma(X)$ measurable, is
  said to be \emph{coherent}\footnote{Note that a coherent pair $(f,s)$ by definition satisfies condition (1) and condition (2) of Theorem (\ref{Dis}) but is not required to satisfy condition (3). } when $f$ is measure preserving, i.e., $q = p \circ f^{-1}$, and $s_y$
is a probability measure on $f^{-1}(y)$ $q$-almost everywhere.  \footnote{Note that $(f,s)$ being coherent is equivalent to $\eta_Y = \Gamma(f) \circ s$.}  If in addition, $p$ is absolutely
  continuous with respect to $s \mon q$, then we say that $(f,s)$ is
  \emph{absolutely coherent}.
\end{defi}

\begin{defi}
The category \Fins{} has
\begin{itemize}
\item \textbf{Objects} : Pairs $(X, p)$ where $X$ is a finite set and $p$ a
probability measure on $X$.  
\item \textbf{Morphisms} : $\operatorname{Hom}(X,Y)$ are all coherent pairs $(f,s)$, $f : X \rightarrow Y$ and $s : Y \rightarrow \Gamma(X)$.
\end{itemize}
We compose arrows $(f,s) : (X,p) \rightarrow (Y,q)$ and
$(g,t) : (Y,q) \rightarrow (Z,m)$ as follows:
$(g,t) \circ (f,s) := (g \circ f, s \monf t)$ where $\monf$ is defined
as \[ (s \monf t)_z(x) = \sum_{y \in Y} t_z(y)s_y(x).\]
\end{defi}

We now leave the finite world for a more general one: the category \Sbs{}.
\pagebreak
\begin{defi}
The category \Sbs{} has
\begin{itemize}
\item \textbf{Objects} : Pairs $(X, p)$ where $X$ is a standard Borel space and $p$ a
probability measure on the Borel subsets of $X$.  
\item \textbf{Morphisms} : $\operatorname{Hom}(X,Y)$ are all coherent pairs
  $(f,s)$, $f : X \rightarrow Y$ and $s : Y \rightarrow \Gamma(X)$. 
\end{itemize}
We compose arrows $(f,s) : (X,p)
  \rightarrow (Y,q)$ and $(g,t) : (Y,q) \rightarrow (Z,m)$ as follows: $(g,t) \circ
  (f,s) := (g \circ f, s \mon t)$.  
\end{defi}
Note that the identity arrow on object $(X,p)$ is $(id_X, \eta_X)$ where $id_X$ is the identity function on $X$. Following the graphical representation from~\cite{Baez14} we represent
composition as follows:
\[
\begin{tikzcd}
(X,p) \arrow[r, bend right, "f", swap] & 
(Y,q)\arrow[l, bend right, squiggly, "s", swap] \arrow[r, bend right, "g", swap] &
(Z,m) \arrow[l, bend right, squiggly, "t", swap]
\end{tikzcd}
\xRightarrow[\operatorname{Composition}]{} 
\begin{tikzcd}
(X,p) \arrow[r, bend right, "g \, \circ f", swap] & 
(Z,m) \arrow[l, bend right, squiggly, "s \mon t", swap]
\end{tikzcd}.
\]

One can think of $f$ as a measurement process from $X$ to $Y$ and of $s$ as a hypothesis about $X$ given an observation in $Y$. We say that a hypothesis $s$ is \emph{optimal}\footnote{For a coherent pair $(f,s)$, asking $s$ to be optimal is equivalent to asking that $(f,s)$ satisfies condition (3) in Theorem (\ref{Dis}) as will be shown in Lemma (\ref{conditional_equal}).} if $p = s \mon q$. We denote by $\mathbf{FP}$ the subcategory of \Sbs{} consisting of the same objects, but with only those morphisms where the hypothesis is optimal.  See ~\cite{Baez11,Baez14} and ~\cite{Leinster} for a discussion of these ideas in the finite case.

\begin{prop}\label{comp_stable}
Given coherent pairs the composition is coherent.  If, in addition, they are
  absolutely coherent, the composition is absolutely coherent.
\end{prop}

\begin{proof} 
We first show that the composition is coherent, i.e., $\monid_Z = \left(\Gamma(g) \circ \Gamma(f)\right) \circ (s \mon t)$.  It is sufficient to show that the following diagram commutes:
\[ \begin{tikzcd}
& \Gamma(Y) \arrow[d, "\Gamma(\monid_{Y})"] \arrow[dl, swap, "\Gamma(s)"]   & Z \arrow[l, swap, "t"] \arrow[dd, "\monid_Z"] \\ 
\Gamma^2(X) \arrow[d,"\average_X"] \arrow[r, swap, "\Gamma^2(f)"] & \Gamma^2(Y) \arrow[d, "\mu_{Y}"]  \\
\Gamma(X) \arrow[r, "\Gamma(f)"] & \Gamma(Y) \arrow[r, "\Gamma(g)"] & \Gamma(Z)
\end{tikzcd}\]

Using the hypothesis that $\monid_Z = \Gamma (g) \circ
t$ and the fact that $\operatorname{Id} = \mu \circ \Gamma(\monid)$, we get that the right-hand square commutes.  The triangle commutes since it is the application of $\Gamma$ to our hypothesis $\monid_Y = \Gamma(f) \circ s$ and the left-hand square commutes because $\mu$ is a natural transformation.  Therefore, the whole diagram commutes and we have thus shown the composition of coherent morphisms is also coherent.

Next, in addition, assume the pairs $(f,s)$ and $(g,t)$ are absolutely coherent.  We show
$p \ll (s \mon t \mon m)$.  By hypotheses, $p \ll s \mon q$ and $q \ll t
\mon m$.  Using Proposition \ref{giry_abs} on $q \ll t \mon m$, we get $s \mon q \ll s \mon t
\mon m$.  By transitivity of $\ll$, we conclude $p \ll (s \mon t \mon m)$.  
\end{proof}

We end this section by defining one more category; this one is due to
Lawvere~\cite{Lawvere73}.  It is just the set $[0,\infty]$ but endowed with categorical
structure.  This allows numerical values associated with morphisms to
be regarded as functors.
\begin{defi}
The category $[0, \infty]$ has
\begin{itemize}
\item \textbf{Objects} : One single object: $\bullet$.
\item \textbf{Morphisms} : For each element $r \in [0,\infty]$, one arrow
  $r : \bullet \rightarrow \bullet$.  
\end{itemize}
Arrow composition is defined as addition in $[0,\infty]$. Consequently, $0$ is the identity arrow.
\end{defi}
This is a remarkable category with monoidal closed structure and many other
interesting properties.

\section{Relative entropy functor}

We recapitulate the definition of the relative entropy functor on \Fins{}
from Baez and Fritz~\cite{Baez14} and then extend it to \Sbs{}.  
\begin{defi}
The relative entropy functor \relentf{} is defined from \Fins{} to \monoid{} as follows:
\begin{itemize}
\item \textbf{On Objects} : It maps every object $(X,p)$ to $\bullet$.
\item \textbf{On Morphisms} : It maps a morphism $(f,s) : (X,p) \rightarrow (Y, q)$ to $S_{fin}(p, s \monf q)$, where \[S_{fin}(p, s \monf q) := \sum_{x\in X} p(x) \ln \left( \frac{p(x)}{(s \monf q)(x)}\right).\]
\end{itemize}
\end{defi}

The convention from now on will be that $\infty \cdot c = c \cdot \infty  =
\infty$ for $0 < c \leq \infty$ and $\infty \cdot 0 = 0 \cdot \infty = 0$.  We extend
\relentf{} from \Fins{} to \Sbs{}.  

\begin{defi}
The relative entropy functor $RE$ is defined from \Sbs{} to \monoid{} as follows:
\begin{itemize}
\item \textbf{On Objects} : It maps every object $(X,p)$ to $\bullet$.
\item \textbf{On Morphisms} : Given a coherent morphism $(f,s) : (X,p) \rightarrow (Y, q)$, if $(f,s)$ is absolutely coherent, then
  $RE((f,s)) = S(p, s \mon q)$ , where
  \[ S(p, s \mon q) := \int_X \log\left(\frac{\dee p}{\dee (s \mon q)}\right) \; \dee p,
  \]  otherwise it is defined as $RE((f,s)) = \infty$.
\end{itemize}
\end{defi}
This quantity is also known as the \emph{Kullback-Leibler divergence}.  

We could have defined our category to have only absolutely coherent
morphisms but it would make the comparison with the finite case more
awkward as the finite case does not assume the morphisms to be absolutely coherent.  The present definition leads to slightly awkward proofs where we
have to consider absolutely coherent pairs and ordinary coherent pairs
separately.
  
Clearly, $RE$ restricts to \relentf{} on \Fins{}.  If $(f,s)$ is
absolutely coherent, then $p$ is absolutely continuous with respect to
$(s \mon q)$ and the Radon-Nikodym derivative is defined.  The relative
entropy is always non-negative~\cite{Kullback51}; this is an easy
consequence of Jensen's inequality.  This shows that $RE$ is defined
everywhere in \Sbs{}.

We will use the following notation occasionally: 
\[ RE \bigg( \!  \! \! 
\begin{tikzcd}
(X,p) \arrow[r, bend right, "f", swap] & 
(Y,q) \arrow[l, bend right, squiggly, "s", swap]
\end{tikzcd} \!\!\! \bigg) := RE((f,s)).
\]

It's easy to see that $RE$ sends the identity arrows of \Sbs{} to $0$---the identity arrow of the unique object $\bullet$ of $[0, \infty]$. Hence, in order to show that $RE$ is indeed a functor, it suffices to show that
\begin{small}
\[
RE\bigg( \!\!\!
\begin{tikzcd}
(X,p) \arrow[r, bend right, "f", swap] & 
(Y,q) \arrow[l, bend right, squiggly, "s", swap] \arrow[r, bend right, "g", swap] &
(Z,m) \arrow[l, bend right, squiggly, "t", swap]
\end{tikzcd} \!\!\! \bigg) =  RE\bigg( \!\!\!
\begin{tikzcd}
(X,p) \arrow[r, bend right, "f", swap] & 
(Y,q) \arrow[l, bend right, squiggly, "s", swap]
\end{tikzcd} \!\!\! \bigg) +  RE\bigg( \!\!\!
\begin{tikzcd}
(Y,q) \arrow[r, bend right, "g", swap] &
(Z,m) \arrow[l, bend right, squiggly, "t", swap] 
\end{tikzcd} \!\!\! \bigg).
\]
\end{small}In order to do so, we will need the following two lemmas.

\begin{lem}\label{conditional_equal}
Given an arrow $(f,s) : (X, p) \rightarrow (Y, q)$ in \Sbs.  Let $\{(s \mon
q)_y\}_{y \in Y}$ be a disintegration
of $(s \mon q)$ along $f$, then  
\[(s \mon q)_y =s_y \text{ $q$-almost everywhere.} \]
\end{lem}

We just have to show that $\{s_y\}_{y \in Y}$ satisfies the three
properties implied by the disintegration theorem.  We prove the third
one; the first two being obvious.  

\noindent \ref{Dis} (iii) \emph{ : For every Borel-measurable function $h :
  X \rightarrow [0,\infty]$, \[\int_X h \; \dee (s \mon q) = \int_{y \in Y}
    \int_{f^{-1}(y)} h \; \dee s_y \; \dee q.\]} 

\begin{proof}
Let's assume as a special case that $h$ is the indicator function for a
measurable set $E \subset X$.  Then, we have \[ \int_X h \; \dee (s \mon q) =
  \int_E \dee (s \mon q)  = (s \mon q)(E) = \int_{y \in Y} s_y(E) \; \dee q = \int_{y
    \in Y}\int_{f^{-1}(y)}h \; \dee s_y \; \dee q.\] 
We have shown that it is true for any indicator function.  By linearity, it
is true for any simple function and then, by the monotone convergence
theorem, it is true for all Borel-measurable functions $h : X \rightarrow [0,\infty]$.
 \end{proof}

\begin{lem} \label{lem_pre_comp}The relative entropy is preserved under
  pre-composition by optimal hypotheses, i.e., for any $(g,t) : (Y,q) \rightarrow (Z,m)$ and $(f,s) : (X, s \mon q) \rightarrow (Y,q)$,
we have \[ RE\bigg( \!\!\!
\begin{tikzcd}
(Y,q) \arrow[r, bend right, "g", swap] &
(Z,m) \arrow[l, bend right, squiggly, "t", swap]
\end{tikzcd}
\!\!\! \bigg) = RE \bigg( \!\!\!
\begin{tikzcd}
(X,s \mon q) \arrow[r, bend right, "f", swap] & 
(Y,q) \arrow[l, bend right, squiggly, "s", swap] \arrow[r, bend right, "g", swap] &
(Z,m) \arrow[l, bend right, squiggly, "t", swap]
\end{tikzcd} \!\!\! \bigg).  \]
\end{lem}
\begin{proof}

\textbf{Case I : $(g,t)$ is absolutely coherent.}  Since $(g,t)$ is
absolutely coherent, so is $(g \circ f, s \mon t)$ by Proposition
\ref{giry_abs}.  Hence, to show $RE(g,t) = RE(g \circ f, s \mon t)$ is to show 
\[  \int_Y \log\left( \frac{\dee q}{\dee (t
      \mon m)}\right) \; \dee q  = \int_X \log\left(\frac{\dee ( s \mon q)}{\dee (s\mon t
      \mon m)}\right) \; \dee(s \mon q).\]  

Because $f$ is measure preserving, it is sufficient to show that the following functions on $X$ \[
  \frac{\dee q}{\dee (t\mon m)} \circ f = \frac{\dee (s\mon q)}{\dee (s \mon t \mon
    m)}\text{   $s \mon t \mon m$-almost everywhere.}\] 

By the Radon-Nikodym theorem, itx is sufficient to show that for any $E
\subset X$ measurable set, we have \[ (s \mon q)(E) = \int_E
  \frac{\dee q}{\dee (t\mon m)}\circ f \; \dee (s \mon t \mon m).\] 
The following calculation establishes the above.
\begin{align} 
&\int_E \frac{\dee q}{\dee (t\mon m)} \circ f \; \dee (s \mon t \mon m) \nonumber  \\
&\qquad\qquad = \int_Y \left( \int_{x \in f^{-1}(y) \cap
  E}\left(\frac{\dee q}{\dee (t\mon m)}\circ f \right)(x)\; \dee (s \mon t \mon m)_y
  \right)\; \dee( t \mon m)  \label{dis}\\  
&\qquad\qquad = \int_Y \frac{\dee q}{\dee (t\mon m)}(y)\left( \int_{f^{-1}(y) \cap
  E} \dee (s \mon t \mon m)_y \right) \; \dee(t \mon m) \label{constant} \\  
&\qquad\qquad = \int_Y \frac{\dee q}{\dee (t\mon m)}(y)\left( \int_{f^{-1}(y) \cap
  E} \dee s_y \right) \; \dee (t \mon m) \label{ce} \\  
&\qquad\qquad = \int_Y \frac{\dee q}{\dee (t\mon m)}(y) s_y(E \cap f^{-1}(y)) \; \dee (t \mon m) \nonumber  \\
&\qquad\qquad = \int_Y \frac{\dee q}{\dee (t\mon m)}(y) s_y(E) \; \dee (t \mon m)\label{ce0} \\
& \qquad\qquad =  \int_Y s_y(E) \; \dee q\label{ce1} \\ 
& \qquad\qquad = (s \mon q)(E) \label{ce2}
\end{align}

We get (\ref{dis}) by applying the
disintegration theorem to $f : \left(X, s \mon t \mon m\right) \rightarrow (Y, t \mon m)$.  The equation
(\ref{constant}) follows by using the fact that
$\frac{\dee q}{\dee (t\mon m)} \circ f$ is constant on $f^{-1}(y)$ for every $y$.
To obtain (\ref{ce}) we apply Lemma \ref{conditional_equal}.  To show
(\ref{ce0}) we use the fact that $s_y$ is a probability measure on 
$f^{-1}(y)$.  We get (\ref{ce1}) by the definition of the Radon-Nikodym derivative and
we finally establish (\ref{ce2}) by the definition of Kleisli composition.

\textbf{Case II : $(g,t)$ is not absolutely coherent.} We have $RE((g,t)) =
\infty$.  We show that $(g \circ f, s \mon t)$ is not absolutely coherent,
i.e., $s \mon q$ is not absolutely continuous with respect to $s \mon t \mon m$.

Since, by hypothesis, $q \ll t \mon m$ doesn't hold,
there exists a measurable set $B \subset Y$ such that $(t \mon
m)(B) = 0$ but $q(B) > 0$.  We argue that $(s \mon t \mon m)(f^{-1}(B)) = 0$ and $(s \mon q)(f^{-1}(B)) > 0$.
 On one hand, we have 
\[ (s \mon t \mon m)(f^{-1}(B)) = \int_B s_y(f^{-1}(B)) \dee (t \mon m) \leq (t
  \mon m)(B) = 0.\] 
But on the other hand, since $f$ is a measure preserving map from $(X, s \mon q)$ to
$(Y,q)$, we have $(s \mon q)(f^{-1}(B)) = q(B) > 0$.
  
Therefore,
\[ RE((g,t)) = \infty = RE((g \circ f, s \mon t)). \qedhere\] 
\end{proof}

\begin{thm}[Functoriality]\label{func}
Given arrows $(f,s): (X,p) \rightarrow (Y,q)$ and $(g,t) : (Y,q)
\rightarrow (Z,m)$, we have
\[RE\left((g,t) \circ (f,s)\right) = RE((g,t)) +  RE((f,s)).\]
\end{thm}

\begin{proof}
Note that by definition, $RE\left((g,t) \circ (f,s)\right) = RE\left((g \circ f,s \mon t)\right)$.  

\textbf{Case I : $(f,s)$ and $(g,t)$ are absolutely coherent.} 
By Proposition \ref{comp_stable}, we have that $(g \circ f,s \mon t)$ is
absolutely coherent.  

\begin{align}
RE\left((g \circ f,s \mon t)\right) \nonumber & 
=  \int_X \log\left(\frac{\dee p}{\dee (s\mon t \mon m)}\right)\; \dee p \nonumber  \\ & 
=  \int_X \log\left(\frac{\dee p}{\dee (s \mon q)} \frac{\dee (s \mon q)}{\dee (s \mon t \mon m)}\right) \; \dee p  \label{chain_rad} \\ &  
= \int_X \log\left(\frac{\dee p}{\dee( s \mon q)}\right) \; \dee p  + \int_X \log\left(\frac{\dee( s \mon q)}{\dee (s\mon t \mon m)}\right) \; \dee p \nonumber \\ &  
= RE((f,s)) + \int_X \log\left(\frac{\dee( s \mon q)}{\dee(s\mon t \mon m)}\right) \; \dee p \nonumber \\ & 
= RE((f,s)) + RE((g,t)) \label{lem_func} 
\end{align}

We get (\ref{chain_rad}) by the chain rule for Radon-Nikodym derivatives and (\ref{lem_func})
by applying Lemma \ref{lem_pre_comp}.

\textbf{Case II : $(g,t)$ is not absolutely coherent.}
We argue that $(g \circ f, s \mon t)$ is not absolutely coherent.
By hypothesis, $q \ll t \mon m$ doesn't hold, so there is a measurable set $B \subset Y$ such that $(t \mon m)(B) = 0$ and $q(B) > 0$.  We show that $(s \mon t \mon m)\left(f^{-1}(B)\right)= 0$ and $p(f^{-1}(B)) > 0$.  On one hand, we have
 \[ (s \mon t \mon m)\left(f^{-1}(B)\right) = \int_B s_y\left(f^{-1}(B)\right) \; \dee (t \mon m) \leq  (t \mon m)(B) = 0,\]
 but on the other hand, we have $p(f^{-1}(B)) = q(B) > 0$.
Therefore \[RE\left((g,t) \circ (f,s)\right) =  \infty = RE((g,t)) + RE((f,s)).\]

\textbf{Case III : $(f,s)$ is not absolutely coherent.}  

This case is not analogous to the previous case since the existence of a measurable set $A \subset X$ such that $(s \mon q)(A) = 0$ and $p(A) > 0$ is surprisingly not enough to conclude that $(s \mon t \mon m)(A) = 0$.

By the hypothesis of $(f,s)$ not being absolutely coherent,  $p \ll s \mon q$ doesn't hold, so there is a measurable set $A \subset X$ such that
$(s \mon q)(A) = 0$ and $p(A) > 0$.  

We partition $A$ into \[ A_{\epsilon} := \{ x \in A \; | \; s_{f(x)}(A) > 0\} \text{  and  } A_{0} := \{x \in A \; | \; s_{f(x)}(A) = 0\} \] and we partition $Y$ into  \[ B_{\epsilon} := \{ y \in Y \; | \; s_y(A) > 0\} \text{  and  } B_{0} := \{y \in Y \; | \; s_{y}(A) = 0\}.  \]

We argue that $(s \mon t \mon m)(A_0) = 0$ and  $p(A_0) > 0$.

Since $A_0 \subset f^{-1}(B_{0})$, $f^{-1}(B_{\epsilon})$ is disjoint from $A_0$, so for all $y \in B_{\epsilon}$ we have $s_y(A_0) = 0$ because their support is disjoint from $A_0$.  On one hand, we thus have  \begin{align*} (s \mon t \mon m)(A_0) &= \int_{Y} s_y(A_0) \; \dee(t \mon m) \\ & = \int_{B_0} s_y(A_0) \; \dee(t \mon m) + \int_{B_{\epsilon}} s_y(A_0) \; \dee(t \mon m) \\ &=  \int_{B_0} s_y(A_0) \; \dee (t \mon m) \\& \leq \int_{B_0} s_y(A) \; \dee (t \mon m) \\ &=  0.  \end{align*}
On the other hand, since we have  $ p(A_0) + p(A_{\epsilon}) = p(A)  > 0$ and
$A_{\epsilon} \subset f^{-1}(B_{\epsilon})$, it suffices to show
$p(f^{-1}(B_{\epsilon})) = 0$  to conclude $p(A_0) > 0$.  

By hypothesis, we have \[(s \mon q)(A) = \int_{B_{0}} s_y(A) \; \dee q
+ \int_{B_{\epsilon}} s_y(A) \; \dee q = 0, \] so $q(B_{\epsilon}) =
0$ and because $f$ is measure preserving, we have $p(f^{-1}(B_{\epsilon}))
= q(B_{\epsilon}) = 0$ as desired.  

So $(g \circ f, s \mon t)$ is not absolutely coherent, hence \[RE\left((g,t) \circ (f,s)\right) =  \infty =  RE((g,t)) + RE((f,s)).\] This completes the proof of this case.
\end{proof}

We have thus shown that $RE$ is a well-defined functor from \Sbs{} to $[0,\infty]$.

\subsection{Convex linearity}
We show below that the relative entropy functor satisfies a convex
linearity property.  In ~\cite{Baez14} convexity looks familiar; here since
we are performing ``large'' sums we have to express it as an integral.
First we define a localized version of the relative entropy.  

Note that Lemma~\ref{conditional_equal} says that $s_y =
(s\mon q)_y$ $q$-almost everywhere.  Thus, in the following there is no
notational clash between the kernel $s_y$ and $(s \mon q)_y$, the later
being the disintegration of $(s \mon q)$ along $f$.  

Given an arrow $(f,s) : (X,p) \rightarrow (Y,q)$ in \Sb{} and a point $y
\in Y$, we denote by $(f,s)_y$, the morphism $(f,s)$ restricted to the pair
of standard Borel spaces $f^{-1}(y)$ and $\{y\}$.  Explicitly, \[ (f,s)_y
  := (\left.f\right|_{f^{-1}(y)},s_y) : (f^{-1}(y),p_y) \longrightarrow
  (\{y\}, \delta_y), \] 
where $\delta_y$ is the one and only probability measure on $\{y\}$.

\begin{defi}
A functor $F$ from \Sbs{} to \monoid{} is \emph{convex linear} if for every
arrow $(f,s) : (X,p) \rightarrow (Y,q)$, we have \[ F\left( (f,s) \right) =
  \int_{Y} F\left( (f,s)_y \right) \; \dee q.\] 
\end{defi}

We will sometimes refer to the relative entropy of $(f,s)_y$ as the \emph{local relative entropy of $(f,s)$ at $y$}.  Before proving that RE is convex linear, we first prove the following lemma.  

\begin{lem}\label{condtional_abs}
Given \[\begin{tikzcd}
(X,p) \arrow[r, "f"] & 
(Y,q)  &
(X,p') \arrow[l,   "f", swap]
\end{tikzcd}\]
where $f$ is a measurable map preserving the measure of both Borel probability measures $p$ and $p'$.  If $p \ll p'$, then $\frac{\dee p_y}{\dee p_y'}$ is defined for $q$-almost every $y$ and \[\frac{\dee p_y}{\dee p_y'} (x) = \frac{\dee p}{ \dee p'}(x)  \; \text{ $p'$-almost everywhere.}\]
 \end{lem}

\begin{proof}
For an arbitrary measurable function $h : X \rightarrow [0, \infty]$, by first applying the Radon-Nikodym theorem and then the
disintegration theorem on the measurable
function $h \frac{\dee p}{\dee p'}$, we get \[\int_X h \; \dee p = \int_X
  h \frac{\dee p}{\dee p'} \; \dee p' = \int_Y \int_{f^{-1}(y)} h \frac{\dee p}{\dee p'} \; \dee p_y' \; \dee q.\]

Hence, for $q$-almost every $y$, we must have
$\frac{\dee p_y}{\dee p_y'} (x) =\frac{\dee p}{ \dee p'}(x)$ $p'$-almost everywhere.
\end{proof}

\begin{thm}[Convex Linearity]\label{convex}
The functor RE is convex linear, i.e., for every arrow $(f,s) : (X,p) \rightarrow (Y,q)$, we have \[ RE((f,s)) = \int_Y RE\left((f,s)_y\right) \; \dee q .\]
\end{thm}
\pagebreak
\begin{proof}

\textbf{Case I : $(f,s)$ is absolutely coherent.} 

We have 
\begin{align}
RE((f,s)) & = \int_X \log \left( \frac{\dee p}{\dee (s \mon q)} \right)\dee p  \nonumber  \\ 
&= \int_Y \int_{f^{-1}(y)} \log \left( \frac{\dee p}{\dee (s\mon q)} \right) \dee p_y \; \dee q \label{co_1}\\  
&= \int_Y \int_{f^{-1}(y)} \log \left( \frac{\dee p_y}{\dee (s\mon q)_y} \right) \; \dee p_y \; \dee q \label{co_2}\\ 
&=  \int_Y RE((f,s)_y) \; \dee q . \nonumber 
\end{align}
We get (\ref{co_1}) by the disintegration theorem and (\ref{co_2})
by applying Lemma~\ref{condtional_abs}.  

\textbf{Case II : $(f,s)$ is not absolutely coherent.} 
By the hypothesis of $(f,s)$ not being absolutely coherent, there is a measurable set $A \subset X$ such that $(s \mon q)(A) = 0$ and $p(A) > 0$.  Applying lemma~\ref{conditional_equal}, on one hand we have 
\[    \int_Y (s \mon q)_y(A) \; \dee q = \int_Y s_y(A) \; \dee q    = (s \mon q) (A)= 0 , \] but on the other
hand we have \[ \int_Y p_y(A)  \; \dee q = p(A) > 0.\]  Hence, the subset of $Y$ on which
$p_y \ll (s \mon q)_y$ doesn't hold contains a set of measure strictly
greater than $0$.  Therefore,  
\[ RE((f,s)) = \infty =  \int_Y RE\left((f,s)_y\right) \; \dee q. \qedhere \]
\end{proof}

\subsection{Lower-semi-continuity}

Recall that a sequence of probability measures $p_n$ converges strongly to $p$, denoted by $p_n \rightarrow p$, if for all measurable set $E$, one has $\lim_{n \rightarrow \infty} p_n(E) = p(E)$.

The singleton set equipped with the trivial measure, which we will denote by $(1,\delta)$, is a weakly terminal object of \Sbs{}, it is weakly terminal in the sense that for every $(X,p)$ there exist a non-unique arrow $(f,s) : (X,p) \to  (1,\delta)$ in \Sbs{}.

\begin{defi}\label{semicontinuous}
A functor $F$ from \Sbs{} to \monoid{} is \emph{lower semi-continuous} if for every arrow $(f,s) : (X,p) \rightarrow (1,\delta)$, whenever $p_n \rightarrow p$ and $s_n \rightarrow s$, then \[ F \bigg( \!\!\!\begin{tikzcd}
(X,p) \arrow[r, bend right, "f", swap] & 
(1, \delta) \arrow[l, bend right, squiggly, "s", swap] 
\end{tikzcd} \!\!\! \bigg)  \leq \
\liminf_{n \rightarrow \infty} F \bigg( \!\!\! 
\begin{tikzcd}
(X,p_n) \arrow[r, bend right, "f", swap] & 
(1,\delta)
 \arrow[l, bend right, squiggly, "s_n", swap]  
\end{tikzcd} \!\!\!\bigg).
\]
\end{defi}

Recall that in ~\cite{Baez14}, lower semicontinuity was defined on $\Fins$ as the following.

\begin{defi}[Baez and Fritz]\label{semicontinuous_finstats}
A functor $F : \Fins \rightarrow [0,\infty]$ is \emph{lower semicontinuous} if for any sequence of morphisms $(f,s_i) : (X,p_i) \rightarrow (Y,q_i)$ that converges\footnote{Where convergence is just pointwise convergence.} to a morphism $(f,s) : (X,p) \rightarrow (Y,q)$, we have \[ F(f,s) \leq \liminf_{i \rightarrow \infty} F(f,s_i).\]
\end{defi}

Recalling that $\mathbf{FP}$ stands for the subcategory of \Sbs{} consisting of the same objects, but with only those morphisms where the hypothesis is optimal. We claim that a lower semi-continuous (as defined in \mbox{Definition \ref{semicontinuous}}) functor $F$ that vanishes on $\mathbf{FP}$ restricts to a lower semi-continuous functor on $\Fins$ (as defined in \mbox{Definition \ref{semicontinuous_finstats}}).  To see this, note that, given a sequence of morphisms $(f,s_i) : (X,p_i) \rightarrow (Y,q_i)$ that converges pointwise to a morphism $(f,s) : (X,p) \rightarrow (Y,q)$, we can recover \[F(f,s) \leq \liminf_{i \rightarrow \infty} F(f,s_i)\]
from\[
\begin{tikzcd}
(X,p_i) \arrow[r, bend right, "f", swap] & 
\left( Y, q_i \right) \arrow[l, bend right, squiggly, "s_i", swap] 
\arrow[r, bend right, "g", swap] & (1,\delta) \arrow[l, bend right, squiggly, "q_i", swap]
\end{tikzcd}
\]
and 
\begin{align*} F(f,s) & =  F(f,s) + \overbrace{F(g,q)}^{0}  = F(g \circ f, s \mon q) \\ & \leq  \liminf_{i \rightarrow \infty} F(g \circ f,s_i \mon q_i)  = \liminf_{i \rightarrow \infty} F(f, s_i)+\liminf_{i \rightarrow \infty}\underbrace{F(g,q_i)}_{0} \\ &= \liminf_{i \rightarrow \infty} F(f,s_i).
\end{align*}
Note that, on finite sets, converging pointwise is equivalent to strong convergence.  

\begin{thm}[Lower semi-continuity]
The functor $RE$ is lower semi-continuous.
\end{thm}
\begin{proof}

Let us denote \[a:=\liminf_{n \rightarrow \infty} RE \bigg( \!\!\! 
\begin{tikzcd}
(X,p_n) \arrow[r, bend right, "f", swap] & 
(1,\delta)
 \arrow[l, bend right, squiggly, "s_n", swap]  
\end{tikzcd} \!\!\!\bigg) .\] If $a = \infty$, then the statement holds automatically, so we assume that $a < \infty$.

By virtue of $a$ being a limit inferior, we can pick a subsequence $\{n_i\}_{i \in \mathbb{N}}$ such that for all $i \in \mathbb{N}$, we have both
\[
RE \bigg( \!\!\! 
\begin{tikzcd}
(X,p_{n_i}) \arrow[r, bend right, "f", swap] & 
(1,\delta)
 \arrow[l, bend right, squiggly, "s_{n_i}", swap]  
\end{tikzcd} \!\!\!\bigg) = \int_X \log\left( \frac{\dee p_{n_i}}{\dee s_{n_i}} \right) \dee p_{n_i} < \infty
\] and \[
\lim_{i\rightarrow \infty} RE \bigg( \!\!\! 
\begin{tikzcd}
(X,p_{n_i}) \arrow[r, bend right, "f", swap] & 
(1,\delta)
 \arrow[l, bend right, squiggly, "s_{n_i}", swap]  
\end{tikzcd} \!\!\!\bigg) = a.\]

Now, instantiating statements (2.4.7) and (2.4.9) from Pinsker~\cite[Section 2.4]{Pinsker60}\footnote{Or equivalently, and perhaps a more accessible reference, Theorem 1 from \cite{posner1975random}.} in our setting, we have
 \[ RE \bigg( \!\!\!\begin{tikzcd}
(X,p) \arrow[r, bend right, "f", swap] & 
(1, \delta) \arrow[l, bend right, squiggly, "s", swap] 
\end{tikzcd} \!\!\! \bigg)  \leq \
\lim_{i \rightarrow \infty} F \bigg( \!\!\! 
\begin{tikzcd}
(X,p_{n_i}) \arrow[r, bend right, "f", swap] & 
(1,\delta)
 \arrow[l, bend right, squiggly, "s_{n_i}", swap]  
\end{tikzcd} \!\!\!\bigg) = a,
\]
as desired.
\end{proof}

\section{Uniqueness}
We now show that the relative entropy is, up to a multiplicative constant,
the unique functor satisfying the conditions established so far.  We first
prove a crucial lemma.  

\begin{lem}\label{lem_strong}
  Let $X$ be a Borel space equipped with probability measures $p$ and $q$,
  if $p \ll q$, then we can find a sequence of simple functions $p_n^*$ on
  $X$ such that for the sequence of probability measures
  $p_n(E) := \int_E p_n^* \; \dee q$, we have that $p_n$ and $p$ agree on
  the elements of the partition on $X$ induced by $p_n^*$ and moreover,
  $p_n \rightarrow p$ strongly.
\end{lem}

\begin{proof}
  We write $I_{n,k}$ for the interval $[k2^{-n}, (k +1) 2^{-n})$ and
  $I_{n, \leq}$ for the interval $[n,\infty)$.  Denote by $K_n$ the index
  set $\{0, 1, \ldots, n2^n -1, \leq \}$ of $k$.  We fix a version
  $\frac{\dee p}{\dee q}$ of the Radon-Nikodym such that
  $\frac{\dee p}{\dee q} < \infty$ everywhere.  We define a family of
  partitions and a family of simple functions as follows:
  \[ X_{n,k} := \left\{ x' \in X \; | \; \frac{\dee p}{\dee q}(x') \in
      I_{n,k}\right\}, \qquad p_n^*(x) :=
    \frac{p\left(X_{n,k}\right)}{q\left(X_{n,k}\right)} \text{ on } x \in
    X_{n,k}.\]

Every function induces a partition on the domain; if moreover the function
is simple, the induced partition is finite.  

We first note that $p_n$ and $p$ agree on the elements of the partition induced by $p_n^*$: \[ p_n(X_{n,k}) =  \int_{X_{n,k}} p_n^* \; \dee q =\int_{X_{n,k}} \frac{p\left(X_{n,k}\right)}{q\left(X_{n,k}\right)} \; \dee q = \frac{p\left(X_{n,k}\right)}{q\left(X_{n,k}\right)}q(X_{n,k}) = p(X_{n,k}).\]

Next, we prove the strong convergence of $p_n \rightarrow p$.  We first show $p_n^* \rightarrow \frac{\dee p}{\dee q}$ pointwise.  Let $x \in X$.  Pick $N$ large enough such that $\frac{\dee p}{\dee q}(x) \leq N$.  For a fixed integer $n \geq N$, there is exactly one $k_n$ for which $x \in X_{n,k_n}$.  On the one hand, we have  $k_n2^{-n} \leq \frac{\dee p }{\dee q}(x) \leq (k_n +1)2^{-n}$ on $X_{n,k_n}$.  But on the other hand, by  integrating over $X_{n,k_n}$ and dividing everything by $q(X_{n,k_n})$, we also have $k_n2^{-n}  \leq \frac{p\left(X_{n,k_n}\right)}{q\left(X_{n,k_n}\right)} \leq (k_n+1)2^{-n}$ on $X_{n,k_n}$.  We thus get pointwise convergence since we have
 \[ \left| p_n^*(x) - \frac{\dee p}{\dee q}(x)\right| = \left| \frac{p\left(X_{n,k_n}\right)}{q\left(X_{n,k_n}\right)} - \frac{\dee p}{\dee q}(x) \right|  \leq 2^{-n} \text{ for any } n \geq N.  \]

From the above inequality and the choice of $N$, we note the following \[ p_n^*(x) \leq \frac{\dee p}{\dee q}(x) + 2^{-n} \leq \frac{\dee p}{\dee q}(x) + 1, \quad \text{for $x$ with $\frac{\dee p}{\dee q}(x) < n$},\] \[ p_n^*(x) = p(X_{n,\leq}) \leq 1 \leq\frac{\dee p}{\dee q}(x) + 1, \quad \text{ for $x$ with $\frac{\dee p}{\dee q}(x) \geq n$}.\]
 So for all $n$, we can bound $p_n^*(x)$ everywhere by the integrable function $g(x) : = \frac{\dee p}{\dee q}(x) + 1$.  Given a measurable set $E \subset X$, we can thus apply Lebesgue's dominated convergence theorem.  We get \[ \lim_{n \rightarrow \infty} p_n(E) = \lim_{n \rightarrow \infty} \int_E p_n^* \; \dee q = \int_E \lim_{n \rightarrow \infty} p_n^* \; \dee q = \int_E \frac{\dee p }{\dee q} \; \dee q = p(E). \qedhere \]
\end{proof}

Before proving uniqueness, we recall the main theorem of Baez and Fritz~\cite{Baez14} on \Fins{}.
\begin{thm}\label{origin}
Suppose that a functor \[F: \Fins{} \rightarrow [0,\infty]\] is lower
semicontinuous, convex linear and vanishes on $\mathbf{FP}$.  Then for some $0 \leq c
\leq \infty$ we have $F(f,s) = cRE_{fin}(f,s)$ for all morphisms $(f,s)$ in
\Fins.  
\end{thm}

We are now ready to extend this characterization to \Sbs{}.
\begin{thm}
Suppose that a functor \[F: \Sbs \rightarrow [0,\infty]\] 
is lower semicontinuous, convex linear and vanishes on $\mathbf{FP}$.  
Then for some $0 \leq c
\leq \infty$ we have $F(f,s) = c RE(f,s)$ for all morphisms.  
\end{thm}
\begin{proof}
  Since $F$ satisfies all the above properties on \Fins, we can apply Theorem \ref{origin} in order to establish that $F= cRE_{fin} = cRE$
  for all morphisms in the subcategory \Fins{}.  We show that $F$ extends
  uniquely to $cRE$ on all morphisms in \Sbs{}.  
  
  By convex linearity of $F$, for an arbitrary morphism $(f,s)$ from $(X,p)$ to $(Y,q)$, we have \[F((f,s)) = \int_Y F((f,s)_y) \; \dee q,\] so $F$
  is totally described by its local relative entropies.  It is thus sufficient to
  show $F = cRE$ on an arbitrary morphism $(f,s) : (X,p) \rightarrow
  (1,\delta)$.  The case where $p$ is not absolutely continuous with respect to $s$ is straightforward, so let us assume $p \ll s$.
  
We apply Lemma \ref{lem_strong} with $p$ and $s$ to get the  family of simple functions $p_n^* $ and the corresponding family of partitions $\{X_{n,k}\}$.  We define $\pi_n$ as the function that maps $x \in X_{n,k'}$ to the element $X_{n,k'} \in \{X_{n,k}\}_{k \in K_n}$.  Denote by $s_{\pi_n}$ the disintegration of $s$ along $\pi_n$ and by $s_n$ the corresponding marginal.  Note that since $p_n$ and $p$ agree on every $X_{n,k}$, $p_n$ is indeed the push-forward of $p$ along $\pi_n$.  Presented as diagrams, we have
 \[  
\begin{tikzcd}
(X,p) \arrow[r, bend right, "\pi_n", swap] & 
\left( \{X_{n,k}\}, p_n \right) \arrow[l, bend right, squiggly, "s_{\pi_n}", swap] 
\arrow[r, bend right, "f_n", swap] & (1,\delta) \arrow[l,bend right, squiggly, "s_n", swap]
\end{tikzcd}
\xRightarrow[ \operatorname{Composition}]{} 
\begin{tikzcd}
(X,p) \arrow[r, bend right, "f", swap] & 
(1,\delta) \arrow[l, bend right, squiggly, "s", swap].
\end{tikzcd}
\]

From the above diagram and the hypothesis that $F$ is a functor, we have the following inequality 
\begin{align}\label{right_side} F((f_n,s_n)) \leq F((f,s)) \text{, for all $n \in \mathbb{N}$}.\end{align}
Note that, on the one hand the disintegration of $p_n$ along $\pi_n$ at the point $X_{n,k'} \in \{X_{n,k}\} $ is given by $p_{n,\pi} : = p_n(\cdot)/p_n(X_{n,k'})$, but on the other hand, for any measurable set $E \subset X$, we also have \[ \sum_{k \in K_n} \left( \int_{X_{n,k}} \mathbbm{1}_E \; \dee p_{n,\pi} \right) s_n(X_{n,k}) = \sum_{k \in K_n} \left(\frac{p_n(E \cap X_{n,k})}{p_n(X_{n,k})}\right) s_n(X_{n,k}) \]
\[  = \sum_{k \in K_n} \left(\frac{s(E \cap X_{n,k})}{s(X_{n,k})}\right) s_n(X_{n,k}) = \sum_{k \in K_n}  s(E \cap X_{n,k}) = s(E).\] 

This means that $p_{n,\pi}$ is the disintegration of $s$ along $\pi_n$.  Presented as diagrams, where we use $f^{p_n}$ instead of $f$ to indicate that the arrow leaves from the object $(X,p_n)$ as opposed to $(X,p)$, we have
 \[  
\begin{tikzcd}
(X,p_n) \arrow[r, bend right, "\pi_n", swap] & 
\left( \{X_{n,k}\}, p_n \right) \arrow[l, bend right, squiggly, "p_{n,\pi}", swap] 
\arrow[r, bend right, "f_n", swap] & (1,\delta) \arrow[l,bend right, squiggly, "s_n", swap]
\end{tikzcd}
\xRightarrow[\operatorname{Composition}]{} 
\begin{tikzcd}
(X,p_n) \arrow[r, bend right, "f^{p_n}", swap] & 
(1,\delta) \arrow[l, bend right, squiggly, "s", swap].
\end{tikzcd}
\]

But since $F$ vanishes on $\mathbf{FP}$, we have $F\left( (\pi_n, p_{n,\pi} ) \right)  = 0$.  Combined with the fact that $F$ is a functor, we get
\begin{align}\label{equ_fp}
F\left((f^{p_n}, s)\right) =  F\left( (\pi_n, p_{n,\pi} ) \right) + F\left( (f_n, s_n)\right) = F\left( (f_n, s_n)\right).
\end{align}

By Lemma \ref{lem_strong}, we know that $p_n \rightarrow p$, in terms of our diagrams we have

\[  
\begin{tikzcd}
(X,p_n) \arrow[r, bend right, "f^{p_n}", swap] & 
(1,\delta) \arrow[l, bend right, squiggly, "s", swap]
\end{tikzcd}
\xRightarrow[\operatorname{Strong \; Convergence}]{} 
\begin{tikzcd}
(X,p) \arrow[r, bend right, "f", swap] & 
(1,\delta) \arrow[l, bend right, squiggly, "s", swap].
\end{tikzcd}
\]
Hence, combining (\ref{equ_fp}) with the lower semicontinuity of $F$, we also have the inequality 
\begin{align}\label{left_side}F((f,s)) \leq \liminf_{n \rightarrow \infty}
  F((f^{p_n},s)) = \liminf_{n \rightarrow \infty}
  F((f_n  ,s_n)) .\end{align} 
Since $(f_n,s_n)$ is in \Fins{}, we must have $F((f_n,s_n)) =
cRE((f_n,s_n))$.
Thus, combining (\ref{right_side}) and (\ref{left_side}), we get that $F((f,s))$ must satisfy
\[\limsup_{n
    \rightarrow \infty} cRE((f_n,s_n)) \leq F((f,s)) \leq \liminf_{n
    \rightarrow \infty} cRE((f_n,s_n)),\] but so does $cRE((f,s))$.  We also have \[\limsup_{n
    \rightarrow \infty} cRE((f_n,s_n)) \leq cRE((f,s)) \leq \liminf_{n
    \rightarrow \infty} cRE((f_n,s_n)).\] 
Therefore $F((f,s)) = cRE((f,s))$, as desired.  \end{proof}


\section{Conclusions and Further Directions}
As promised, we have given a categorial characterization of relative
entropy on standard Borel spaces.  This greatly broadens the scope of the original
work by Baez et al.\@~\cite{Baez11,Baez14}.  However, the main 
motivation is to study the role of entropy arguments in machine learning.
These appear in various ad-hoc ways in machine learning but with the appearance of
the recent work by Danos and his
co-workers~\cite{Danos15,Clerc17,Dahlqvist16} we feel that we have the
prospect of a mathematically well-defined framework on which to understand
Bayesian inversion and its interplay with entropy.  The most recent paper
in this series~\cite{Clerc17} adopts a point-free approach introduced
in~\cite{Chaput09,Chaput14}.  It would be interesting to extend our
definitions to a point-free situation.

\bibliographystyle{alphaurl}
\bibliography{Relative_Entropy_LMCS.bib}

\begin{thebibliography}{CDDG17}

\bibitem[Ash72]{Ash72}
R.~B. Ash.
\newblock {\em Real Analysis and Probability}.
\newblock Academic Press, 1972.

\bibitem[BF14]{Baez14}
John~C. Baez and Tobias Fritz.
\newblock A bayesian characterization of relative entropy.
\newblock {\em Theory and Applications of Categories}, 29(16):422--456, 2014.

\bibitem[BFL11]{Baez11}
John~C. Baez, Tobias Fritz, and Tom Leinster.
\newblock A characterization of entropy in terms of information loss.
\newblock {\em Entropy}, 13(11):1945--1957, 2011.

\bibitem[Bil95]{Billingsley95}
P.~Billingsley.
\newblock {\em Probability and Measure}.
\newblock Wiley-Interscience, 1995.

\bibitem[CDDG17]{Clerc17}
Florence Clerc, Vincent Danos, Fredrik Dahlqvist, and Ilias Garnier.
\newblock Pointless learning.
\newblock In {\em Proceedings of FoSSaCS 2017}, 2017.

\bibitem[CDPP09]{Chaput09}
Philippe Chaput, Vincent Danos, Prakash Panangaden, and Gordon Plotkin.
\newblock Approximating {M}arkov processes by averaging.
\newblock In {\em Proceedings of the 37th International Colloquium On Automata
  Languages And Programming (ICALP)}, volume 5556 of {\em Lecture Notes In
  Computer Science}, pages 127--138, 2009.

\bibitem[CDPP14]{Chaput14}
Philippe Chaput, Vincent Danos, Prakash Panangaden, and Gordon Plotkin.
\newblock Approximating {M}arkov processes by averaging.
\newblock {\em J. ACM}, 61(1):5:1--5:45, January 2014.
\newblock \href {https://doi.org/10.1145/2537948} {\path{doi:10.1145/2537948}}.

\bibitem[DDG16]{Dahlqvist16}
Fredrik Dahlqvist, Vincent Danos, and Ilias Garnier.
\newblock Giry and the machine.
\newblock {\em Electronic Notes in Theoretical Computer Science}, 325:85--110,
  2016.

\bibitem[DG15]{Danos15}
Vincent Danos and Ilias Garnier.
\newblock Dirichlet is natural.
\newblock {\em Electronic Notes in Theoretical Computer Science}, 319:137--164,
  2015.

\bibitem[DSDG18]{dahlqvist2018borel}
Fredrik Dahlqvist, Alexandra Silva, Vincent Danos, and Ilias Garnier.
\newblock Borel kernels and their approximation, categorically.
\newblock {\em Electronic Notes in Theoretical Computer Science}, 341:91--119,
  2018.

\bibitem[Dud89]{Dudley89}
R.~M. Dudley.
\newblock {\em Real Analysis and Probability}.
\newblock Wadsworth and Brookes/Cole, 1989.

\bibitem[Gir81]{Giry81}
M.~Giry.
\newblock A categorical approach to probability theory.
\newblock In B.~Banaschewski, editor, {\em Categorical Aspects of Topology and
  Analysis}, number 915 in Lecture Notes In Mathematics, pages 68--85.
  Springer-Verlag, 1981.

\bibitem[KL51]{Kullback51}
Solomon Kullback and Richard~A. Leibler.
\newblock On information and sufficiency.
\newblock {\em The annals of mathematical statistics}, 22(1):79--86, 1951.

\bibitem[Law64]{Lawvere64a}
F.~W. Lawvere.
\newblock The category of probabilistic mappings.
\newblock Unpublished typescript., 1964.

\bibitem[Law73]{Lawvere73}
F.~William Lawvere.
\newblock Metric spaces, generalized logic and closed categories.
\newblock {\em Rend. Sem. Mat. Fis. Milano}, 43(1):135--166, 1973.

\bibitem[Lei]{Leinster}
Tom Leinster.
\newblock An operadic introduction to entropy.
\newblock n-category cafe.

\bibitem[Pin60]{Pinsker60}
Mark~S Pinsker.
\newblock Information and information stability of random variables and
  processes.
\newblock {\em Holden-Day 1964}, 1960.

\bibitem[Pos75]{posner1975random}
Edward Posner.
\newblock Random coding strategies for minimum entropy.
\newblock {\em IEEE Transactions on Information Theory}, 21(4):388--391, 1975.

\bibitem[Rok49]{Rokhlin49}
Vladimir~Abramovich Rokhlin.
\newblock On the fundamental ideas of measure theory.
\newblock {\em Matematicheskii Sbornik}, 67(1):107--150, 1949.

\end{thebibliography}

\end{document}